\documentclass[12pt]{article}

\usepackage{amssymb,amsmath,amsfonts,eurosym,geometry,ulem,graphicx,caption,color,setspace,sectsty,comment,footmisc,caption,natbib,pdflscape,subfigure,array,hyperref,tikz}
\usepackage{natbib}
\usepackage{url}

\newtheorem{proposition}{Proposition}
\newtheorem{remark}{Remark}

\newenvironment{proof}[1][Proof]{\noindent\textbf{#1.} }{\ \rule{0.5em}{0.5em}}

\newcolumntype{L}[1]{>{\raggedright\let\newline\\arraybackslash\hspace{0pt}}m{#1}}
\newcolumntype{C}[1]{>{\centering\let\newline\\arraybackslash\hspace{0pt}}m{#1}}
\newcolumntype{R}[1]{>{\raggedleft\let\newline\\arraybackslash\hspace{0pt}}m{#1}}

\usepackage{natbib}

\title{A Solomonic Solution to Ownership Disputes: An Application to Blockchain Front-Running}

\author{Joshua S. Gans and Richard Holden*}

\begin{document}

\maketitle

\begin{abstract}
    Blockchain front-running involves multiple agents, other than the legitimate agent, claiming a payment from performing a contract. It arises because of the public nature of blockchain transactions and potential network congestion. This paper notes that disputes over payments are similar to classic ownership disputes (such as King Solomon's dilemma). We propose a simultaneous report mechanism that resolves Solomon's dilemma (using only ordinal preference information) and also eliminates blockchain front-running. In each case, the mechanism relies on threats to remove ownership from all claimants and preferences from the legitimate claimant over allocations to other agents.    
    
    \textit{Keywords}: subgame perfect implementation, blockchain, front-running, mechanism design, ownership
    
\end{abstract}

\vfill
\begin{footnotesize}
\noindent * (Gans) Rotman School of Management, University of Toronto and NBER; (Holden) Department of Economics, UNSW Business School. 

\noindent** All correspondence to joshua.gans@utoronto.ca. The latest version of this paper is available at \url{https://www.joshuagans.com/} and \url{https://github.com/solomonic-mechanism}. Thanks to Ethan Buchman, Scott Kominers and DJ Thornton for useful discussions and Raphael Mu for excellent research assistance. 
\end{footnotesize}
\newpage

\section{Introduction}

Front-running has become a serious issue for smart contracts in blockchain ecosystems; threatening to completely undermine its potential.\footnote{See \cite{catalini2020some}, \cite{gans2019fine} and \cite{holden_malani_2021} for overviews of the economic potential of the blockchain to solve contracting issues.} The problem is straightforward. When a contract is placed on a blockchain such as Ethereum, there is a performance obligation on one party that, when achieved, triggers a payment in tokens from another party. Sometimes these contracts are open offers -- such as a bounty or reward. When performance occurs, the intended payee sends a message to the payor that is akin to an invoice for payment together with evidence that the obligation was met. Being the blockchain, these messages are public prior to being committed to a block. Also, as there is potential congestion on the network, a message is sent with a delay depending upon the transaction fee nominated by the payee. In the intervening time, front-runners, or bots programmed to front-run, see the message and can resend it, substituting in their own address for payment and a higher transaction fee to achieve priority (\cite{daian2019flash}, \cite{eskandari2019sok}). The payor's account for that contract is then drained of tokens before the intended payee can be paid.\footnote{The total value of tokens gained in this manner is estimated at almost \$1 billion since January 2020 (\url{https://explore.flashbots.net/}) although the vast majority of that is from arbitrage front-running rather than liquidition front-running which is the focus of this paper. See also, \cite{ferreira2021frontrunner}.}  

While such front-running is akin to the leap-frogging activities of high frequency traders,\footnote{This occurs were a large trade is placed and bots are able to trade in the market ahead of that trade and exploit arbitrage opportunities. This happens for cryptocurrencies on the blockchain as well using a technique called `insertion' to front-run high value transactions; \cite{ferreira2021frontrunner}. However, this is not the type of front-running considered in this paper.} in this case, it threatens to undermine the ability to offer smart contracts on any blockchain system.\footnote{See \cite{robinson2020ethereum}. The problem was first identified in 2014 in a Reddit post from \href{https://www.reddit.com/r/ethereum/comments/2d84yv/miners_frontrunning/}{pmcgoohan}; see \cite{stankovic2021} for an overview. It is also possible that this activity could undermine the consensus layer of blockchains through front-running on past blocks using a time bandit attack; \cite{daian2019flash}.} Fearing non-payment, a contract payee may not perform or enter into a contract at all. This harms both parties and will likely stifle the development of smart contracts and the ensuing gains from trade. While some solutions involving encrypting messages have been posited, these can only potentially assist in some bilateral contracts between known and identifiable parties (\cite{copeland2021}) unless implemented at a platform level. Other solutions involving increasing the transparency of ``front-running" races do not actually resolve the problem and merely place the payee on a more level playing field than front-running bots.\footnote{\url{https://ethresear.ch/t/flashbots-frontrunning-the-mev-crisis/8251} and the critique by Ed Felten \url{https://medium.com/offchainlabs/meva-what-is-it-good-for-de8a96c0e67c} Such auctions may also reduce the congestion effects generated by front-running; \href{https://medium.com/@VitalikButerin/i-feel-like-this-post-is-addressing-an-argument-that-isnt-the-actual-argument-that-mev-auction-b3c5e8fc1021}{Buterin response}.}

In this paper, we provide and examine a mechanism designed to resolve ownership disputes that fall into a specific class; of which front-running of the type describe here is an example. Another famous example is the biblical dispute heard by King Solomon. The class of disputes have the following characteristics:
\begin{enumerate}
    \item The legitimate claimant is part of the set of agents making an ownership claim.
    \item Legitimate and illegitimate claimants know if they are legitimate or not.
    \item Legitimate and illegitimate claimants have different preferences over who, other than themselves, are allocated ownership.
\end{enumerate}
In the case of Solomon's adjudication over who was the true mother of a baby, it was known that the true mother was one of the set of two claimants, each claimant knew their own status and, as we will discuss, it was a feature of the story that the true mother had different preferences than the other agent over what happened to the baby if ownership was not allocated to them. For blockchain front-running, the nature of the problem necessitates the legitimate claimant being part of the relevant claimant set, claimants knowing their own status and illegitimate claimants being indifferent was to other outcomes which may not be the case for the legitimate claimant.

The mechanism we deploy is a simple special case of the Simultaneous Report (SR) mechanism developed by \cite{chen2018getting} that itself is a simplification of mechanisms explored by \cite{moore1988subgame} and \cite{moore1992implementation}.\footnote{It is, however, distinct from the divided ownership processes examined by \cite{ayres1994solomonic} as it envisages a solution outcome whereby ownership is not divided.}

This paper proceeds as follows. In the next section, we revisit King Solomon's dilemma as a warm-up exercise but, in the process, show how the SR mechanism provides a more robust solution with attractive solutions compared to everything else proposed over the last three millennia. Section 3 then sets up the front-running problem and provides a mechanism (a Solomonic clause), implementable on blockchains, that resolves it entirely. Section 4 concludes.

\section{King Solomon's Baby-Ownership Dispute}

The story of Solomon's dilemma comes the First Book of Kings, Chapter 3, beginning at the 16th verse, and goes as follows.

\begin{quote}

Then came there two women, that were harlots, unto the king, and stood before him. And the one woman said: 'Oh, my lord, I and this woman dwell in one house; and I was delivered of a child with her in the house.
 And it came to pass the third day after I was delivered, that this woman was delivered also; and we were together; there was no stranger with us in the house, save we two in the house. And this woman's child died in the night; because she overlay it. And she arose at midnight, and took my son from beside me, while thy handmaid slept, and laid it in her bosom, and laid her dead child in my bosom. And when I rose in the morning to give my child suck, behold, it was dead; but when I had looked well at it in the morning, behold, it was not my son, whom I did bear.' And the other woman said: `Nay; but the living is my son, and the dead is thy son.' And this said: `No; but the dead is thy son, and the living is my son.' Thus they spoke before the king.
 
 Then said the king: `The one saith: This is my son that liveth, and thy son is the dead; and the other saith: Nay; but thy son is the dead, and my son is the living.' And the king said: `Fetch me a sword.' And they brought a sword before the king. And the king said: 'Divide the living child in two, and give half to the one, and half to the other.'
 
 Then spoke the woman whose the living child was unto the king, for her heart yearned upon her son, and she said: `Oh, my lord, give her the living child, and in no wise slay it.' But the other said: `It shall be neither mine nor thine; divide it.' Then the king answered and said: `Give her the living child, and in no wise slay it: she is the mother thereof.' And all Israel heard of the judgment which the king had judged; and they feared the king; for they saw that the wisdom of God was in him, to do justice.
\end{quote}

\noindent Game theorists have considered many mechanisms to solve the general problem inspired by this biblical account. The literature has focused on mechanisms whereby challenge stages involve bids or second-price auctions (see \cite{glazer1989efficient}, \cite{perry1999general}, \cite{olszewski2003simple}, \cite{qin2009make}, \cite{mihara2012second}). This means that they rely on an assumption that the true mother has a monetary equivalent value (or cardinal utility) greater than the other woman. As observed by \cite{guha2014reinterpreting}, this may not be a reasonable assumption and there is no biblical reference that might lead to that fact. It may well be that the other woman, distraught or misguided, wants the baby more. In any case, these models are often paired with an assumption that this utility is quasi-linear with no wealth effects (\cite{moore1992implementation}) which again stretches their credibility as a solution to King Solomon's dilemma. 

Instead, we propose a mechanism that reveals truthful outcomes while also not imposing additional costs on the true mother, resolving the dispute and requiring the mechanism designer to have only certain ordinal beliefs regarding preferences of each agent and not being able to compare utilities directly. Moreover, the mechanism we offer does not require a non-credible threat (such as killing an innocent child). Our focus instead is on the preference of the other woman who said in Kings, ``It shall be neither mine nor thine; divide it." We interpret this as an indifference relation as to what happens if the baby is not allocated to her. 

\subsection{Model Setup}

There are two women -- Anna ($a$) and Bess ($b$). There are two potential states, $S$, of the world $\{\alpha, \beta\}$ where, under state $\alpha$, $a$ is the true mother and, under state $\beta$, $b$ is the true mother. Both agents, $a$ and $b$, know the true state of the world which is hidden information to everyone else. 

There are three potential outcomes:
\begin{enumerate}
    \item ($A$) Allocate the baby to $a$;
    \item ($B$) Allocate the baby to $b$;
    \item ($C$) Allocate the baby to a third party;
\end{enumerate}
In the biblical account, the third outcome was to kill the baby. \cite{moore1992implementation} adds a fourth outcome where the baby and both agents are killed. We do not consider any fatal outcomes here so that everyone lives although we will have an option to add a punishment of arbitrary size to either $a$ or $b$. 

The mechanism designer possesses the following information regarding agent preferences:
\begin{itemize}
    \item Each agent strictly prefers any given outcome without a punishment to one where they are punished.
    \item Each agent strictly prefers the outcome where they are allocated the baby to an outcome where they are not allocated the baby.
    \item The true mother strictly prefers the baby to be allocated to a third party rather than the other agent.
    \item The other agent is indifferent as to whether the baby is allocated to the true agent or a third party.
\end{itemize}
This third condition on preferences means that the true mother's preferences are dependent on the identity of who is allocated the baby beyond themselves. The rationale here is that the true mother has a preference for their baby to be raised better and, due to the dispute with an agent they know was willing to lie, would prefer $C$ to the allocation of the baby to that agent. Thus, if the state is $\alpha$, then $a$ has the following preference relation: $A \succ C \succ B$ while $b$ has a preference that $B \succ A \sim C$. This last preference relation is inspired by the biblical account that the other woman did not care whether the baby was killed or not. Finally, for any outcome $X \in \{A,B,C\}$ that is paired with a fine, $-F$, we assume that $X \succ_a (X,-F)$ and $X \succ_b (X,-F)$ for all $F>0$. 

\subsection{Proposed Mechanism}

The mechanism we propose is a special case of the Simultaneous Report (SR) Mechanism developed by \cite{chen2018getting}.\footnote{The SR mechanism is a simplification of the multi-stage mechanisms explored by \cite{moore1988subgame}.} Consider the following mechanism:
\begin{enumerate}
    \item One agent is randomly chosen to be the \textit{proposer}, $p$, and the other, $r$, is chosen to be the responder.
    \item The proposer and responder make claims, $M_p, M_r \in \{\alpha,\beta\}$, respectively
    \item If $M_p = M_r$, then the outcome is $A$ if $M_r = \alpha$ and $B$ if $M_r = \beta$.
    \item If $M_p \neq M_r$, then the challenge stage begins with both $p$ and $r$ being fined, $F > 0$. 
\end{enumerate}

\noindent The \textbf{challenge stage} involves:
\begin{enumerate}
    \item $p$ send a new message $M^C_p$ based on knowledge that there is a disagreement.
    \item If $M^C_p = M_r$, then the outcome is $A$ if $M_r = \alpha$ and $B$ if $M_r = \beta$ and $r$ is refunded $F$.
    \item If $M^C_p \neq M_r$, then $C$ is implemented.
\end{enumerate}
Given this, we can prove the following:
\begin{proposition}
The unique subgame perfect equilibrium outcome for the mechanism is $M_p=M_r=S$.
\end{proposition}

\begin{proof}
Without loss in generality, suppose that $S = \alpha$. Consider a challenge stage that has arisen. There are two cases to consider:
\begin{enumerate}
    \item ($p = a$) Assume that $M_r = \beta$. If $M^C_p = \beta$, then $a$ receives, net of any fine, a payoff associated with $B$. If $M^C_p = \alpha$, then $a$ receives, net of any fine, a payoff associated with $C$. Since $a$ is the true mother, $C \succ B$ and so $M^C_p = \alpha$. Moving to the first stage, $b$, as responder, obtains the payoff associated with $C$ plus a fine if $M_r = \beta$ and receives the payoff associated with $A$ if $M_r = \alpha$. So long as $A \succ_b (C,-F)$, $b$ sets $M_r = \alpha$. In the first stage, it is easy to see that $a$ has a weakly dominant strategy to set $M_p = \alpha$ as they either receive the payoff associated with $A$ or $(C,-F)$ which, for $F$ sufficiently small, $a$ prefers to $B$ which could arise if $M_p = \beta$.
    \item ($p = b$) Assume that $M_r = \alpha$. If $M^C_p = \alpha$, then $b$ receives, net of any fine, a payoff associated with $A$. If $M^C_p = \beta$, then $b$ receives, net of any fine, a payoff associated with $C$. Since $a$ is the true mother, $C \sim_b A$ and so that $b$ could choose $M^C_p = \alpha$ or $M^C_p = \beta$ and receive the same utility if the latter results in outcome $C$. Moving to the first stage, $a$, as responder, obtains the payoff associated with $C$ plus a fine if $M_r = \alpha$ and receives the payoff associated with $B$ if $M_r = \beta$. So long as $(C,-F) \succ_a B$, $a$ sets $M_r = \alpha$. In the first stage, it is easy to see that $b$ has a weakly dominant strategy to set $M_p = \alpha$ as they either receive the payoff associated with $A$ which $b$ prefers to $(C,-F) \big(\sim_b (A,-F)\big)$ which could arise if $M_p = \beta$. 
\end{enumerate}
Thus, the mechanism results in truthful revelation.
\end{proof}
\\

\noindent Note here that the fine, $F$, can be set arbitrarily small and still generate truthful revelation in the first stage without the need to actually impose the fine. This sets this mechanism apart from bidding mechanisms that rely on potentially sizeable fines or on mechanisms such as that proposed by \cite{moore1992implementation} that relies on a threat involving maximal utility loss. Both of these may not be credibly implemented even by a despotic, autocratic mechanism designer. 

Interestingly, it is the assumption on agents' rankings of outcomes other than their bliss outcomes that drives this result. The true mother, when faced with giving the baby to a third party or the other woman, will choose the third party. This prevents consensus between the women on the state and forces the other woman to claim the true outcome in order to avoid a fine.

Also important is the fact that, for the other woman, the outcome where the true mother receives the baby is strictly preferred to the third party outcome along with a fine. Our assumption that the other woman is indifferent as to the baby's allocation other than themselves and does not prefer outcomes with fines, generates this ranking but it is the ranking that drives the mechanism. If, for instance, the other woman held some malice towards the true mother (as is examined by \cite{guha2014reinterpreting}), then without a sufficiently large $F$, this preference relation may not hold and this mechanism would not result in truthful revelation.\footnote{\cite{guha2014reinterpreting} examines a mechanism that has a similar consensus followed by potential bad outcomes quality as ours proposed here. However, because of the need to impose large fines and sufficiently bad outcomes, this results in potential wealth effects and other issues leading back to mechanisms that impose assumptions on cardinal utilities. By contrast, under our assumptions here, neither large fines nor catastrophic outcomes are required and all assumptions are regarding preference orderings.}

\section{Blockchain Front-Running}

We now turn to apply a variant of the above mechanism to the ownership dispute inherent in blockchain front-running. It shares with the above analysis that the outcome from persistent disputes is that neither party receives ownership and that one party has strict preferences over ownership allocations they are not part of.

\subsection{The Blockchain Contracting Problem}

Our unit of analysis is a given contract. That contract comprises certain performance obligations whose performance can be verified by party $A$ sending a message $M_A = \{\alpha,E\}$ where $\alpha$ is $A$'s wallet address and $E$ is verifiable evidence of performance to the network as a transaction. That transaction is then confirmed to a block and recorded on the blockchain. At that point, any payment, $T$, triggered by the receipt of $M_A$ involves $T$ in tokens being transferred to $A$'s wallet. Note that any agent, $i$, sending a message that is confirmed to a block specifies and pays a transaction fee, $f_i > 0$.

Front-running arises when $B$ observes $M_A$ as it is broadcast to the network but before it is confirmed to a block and $B$ chooses to send a message $M_B = \{\beta,E\}$ to the network. If $M_B$ is confirmed to a block ahead of $M_A$, then $T$ is automatically sent to $B$'s wallet and $A$ receives no payment. This could arise if $M_B$ is confirmed to a block earlier than the block $M_A$ is confirmed on or if it is confirmed to the same block with an earlier order amongst transactions in that block.

Given that $M_A$ is broadcast first, how could $M_B$ be recorded on the blockchain with an earlier time-stamp? Note, first, the messages are initially broadcast to a mempool. Those transactions are then validated and confirmed by miners or validating nodes who are responsible for ordering the transactions.\footnote{Miners are responsible in proof-of-work protocols while validating nodes are responsible in proof-of-stake protocols.} All valid transactions are recorded on the blockchain at which point the transaction with the earliest timestamp will trigger the contracted actions. Miners or validating nodes will then choose the order of transactions. On the Ethereum blockchain, miners will try and maximise transaction fee revenue by prioritizing transaction recording based on the transaction fee bids (or offered `gas') that accompany a message. Thus, $M_B$ can, by offering to pay a higher transaction fee, be ordered ahead of $M_A$ in a block. Of note is the fact that, because miners have the power to order transactions, to the extent that transaction ordering matters, the ability to earn payments based on ordering power has been termed the \textit{miner-extractable value} (or MEV) (see \cite{daian2019flash}).\footnote{For a demonstration of a smart contracting being front-run in this manner see Scott Bigelow, \href{https://youtu.be/UZ-NNd6yjFM}{``How To Get Front-Run on Ethereum mainnet"}, \href{https://youtu.be/UZ-NNd6yjFM}{YouTube} June 17, 2020.}

It is useful to illustrate the severity of this issue for contracting. Suppose that a party contracts $A$ to provide a service using a contract recorded on the blockchain. If $A$ performs the contract, assume that it costs them, $c>0$, to do so and $T$ will be paid if evidence of performance is submitted. In the absence of front-running, $A$ sends a message, $M_A$ on the blockchain for a fee $f_A$ that can be arbitrarily small and ends up with a payoff of $T-c-f_A$ which is assumed to be positive. 

Suppose now that front-running can occur. If $B$ sends a message $M_B$ for a fee of $f_B$ they can potentially earn $T$. If $f_A = f_B$, then $B$'s expected payoff is $\frac{1}{2}T-f_B$ and $A$'s falls to $\frac{1}{2}T-c-f_A$. In effect, the payment to $A$ is taxed at 50 percent assuming there is only one front-runner. If there are more than one, the effective tax is higher.

This analysis, however, assumes that transaction fees are fixed. However, these are chosen by agents recording transactions on the blockchain. Typically, if $f_B > f_A$, $B$ would be recorded at an earlier time-stamp and their payoff becomes $T-f_B$ while $A$'s drops to $-c-f_A$. In reality, $A$ and $B$ choose their fees as part of a first-price, sealed bid, all-pay auction for priority. Here, $B$ will choose a fee up to $f_B = T$ requiring $A$ to exceed that to achieve priority; something that is not worthwhile. Having both post fees equal to $T$ is not an equilibrium outcome if priority is then randomly assigned. Instead, one sets a fee at $T$ while the other sets an arbitrarily low fee or does not choose to send a message. Given this, either $A$ sets $f_A = T$ and earns a payoff of $-c$ or it sets a low or zero fee and earns the same payoff. Under these conditions, $A$ chooses not to incur $c$ and perform the contract regardless of how high $T$ is.

Under these conditions, contracts that require settlement on the blockchain will not arise in equilibrium. Various solutions have been proposed to mitigate such issues. These have included auctions to make priority a more transparent process (\cite{daian2019flash}, \cite{buterin_2021}). However, these auctions, do not prevent front-runners from participating and that competition still immiserizes contract safety as outlined above. A second solution involves adjusting blockchain protocols to improve time-stamping. However, as there are always lags of some kind achieving this is difficult. A third set of solutions involves encrypting messages until they are confirmed on the blockchain (\cite{aune2018footprints}). This can resolve this problem but it requires implementation at the blockchain protocol level, encrypting all messages which is computationally expensive, and moves away from the public nature of blockchain interactions.   

Compared with contracting outside of the blockchain, the reason why such front-running is a threat is because there is no proof of identity required for payment. This is by design as the privacy of parties on the blockchain is a feature. Thus, blockchain smart-contract systems are characterised by contracts specify performance objectives but not the identity of those performing them. This allows anonymity in payments to be preserved. If the contract specified that following performance, payments would be made to $A$'s wallet ($\alpha$) specifically, front-running could not occur as $M_B$ would not trigger a payment to $\beta$. However, anonymity means that the addressee for payment can be substituted without altering the contract. Our examination here is made on the basis that this blockchain feature needs to be preserved. 

We make the following assumption with regard to agent preferences. The agent who has actually performed the contractual obligation and broadcasts a message of that performance earns $T$ in utility if they receive payment for that performance, $-\theta$ if that payment goes to an illegitimate claimant and $0$ otherwise. (This is analogous to the assumption that the third option, $C$, is preferred by the true claimant in the Solomon example.) 

\subsection{The need to discretize time}

In the literature on front-running in financial markets, one of the proposed solutions was to change time on an exchange from continuous to discrete time (\cite{budish2015high}). In order to operate a mechanism involving multiple agents, to resolve front-running on blockchains we must similarly discretize time. This is done by the smart contract proposing a time period counted from the time a first message $M_i$ is recorded on the blockchain during which all such messages are collected and the mechanism we propose is run on them. The length of the time period, let's call it $\Delta$, is a parameter that can be chosen.\footnote{$\Delta$ can be measured in time units or in blocks with the first block being the one where claim(s) are first confirmed.} Increasing $\Delta$ reduces the need for claimants to pay higher transaction fees so as to participate in the mechanism but also results in a delay in payment. If there is a single message received during $\Delta$, there is no ownership dispute over the payment and the payment is made to the wallet addressee. If there is more than one message received, there is dispute and the mechanism we propose is run to resolve the dispute immediately upon $\Delta$'s end.

\subsection{The single legitimate claimant case}

While our mechanism applies for an arbitrary number of ownership claimants to $T$, initially we assume that there are two claimants, $A$ (the legitimate claimant) and $B$ (the would-be front runner or illegitimate claimant). Each claimant knows their own status but this is known to the mechanism designer. The designer's goal is that the payment only be made to the true claimant. 

The process is initiated as soon as a claim is validated and confirmed on the blockchain. Consider the following mechanism:
\begin{enumerate}
    \item If, in a time period, $\Delta$, there is a single message, $\{i,E\}$, send $T$ to $i$.
    \item If, in a time period, $\Delta$, there are two messages, $M_A = \{\alpha,E\}$ and $M_B = \{\beta,E\}$, the challenge stage begins. (Figure \ref{fig:mempool} illustrates the process by which claims are assembled on a blockchain).
\end{enumerate}
\noindent The \textbf{challenge stage} involves:
\begin{enumerate}
    \item One agent is chosen at random and given the opportunity to withdraw their claim.
    \item If the claim is withdrawn, the other agent is paid $T$.
    \item If the claim is asserted, $T$ is paid to a third party (or, equivalently, the tokens burned) and the contract is nullified.
\end{enumerate}

\begin{figure}
    \centering
  \includegraphics[scale=0.38]{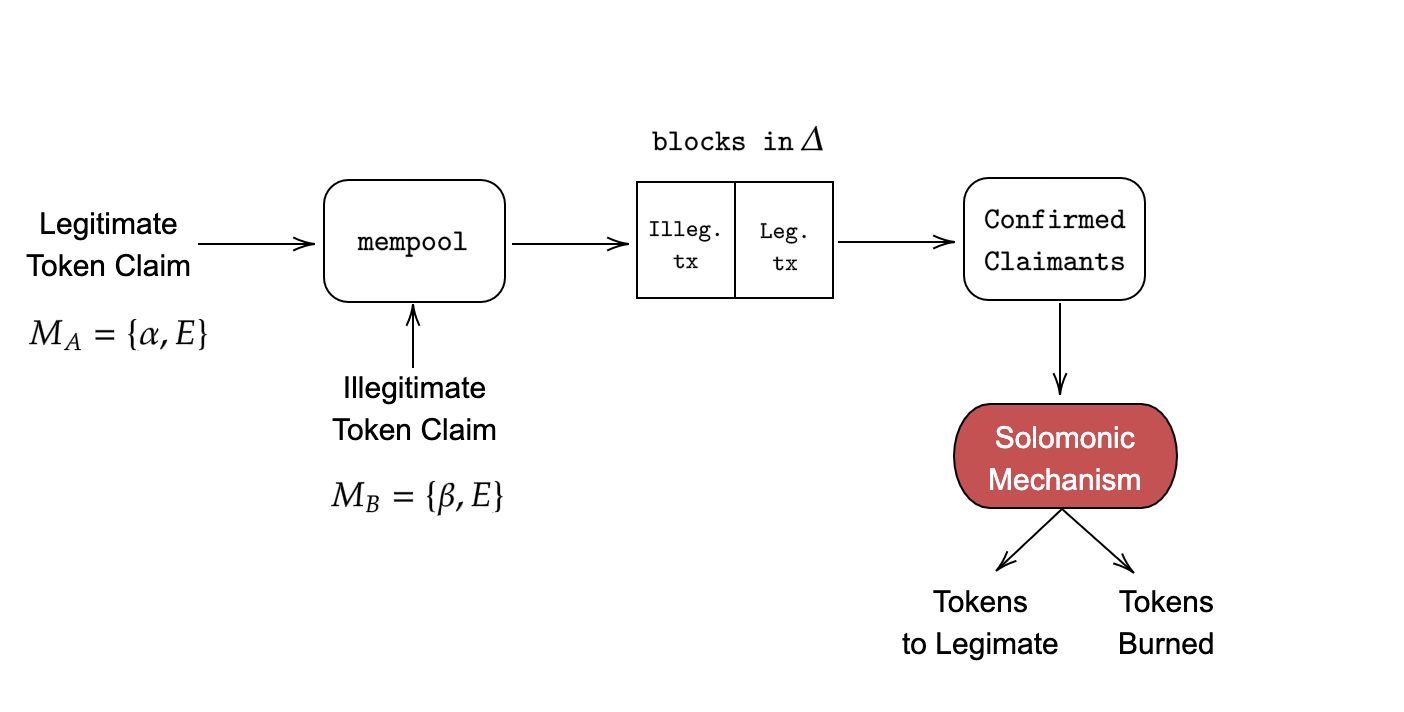}
    \caption{\textbf{Assembling Competing Claims}}
    \label{fig:mempool}
\end{figure}

\noindent Thus, as is depicted in Figure \ref{fig:mempool}, the legitimate claimant first broadcasts a message to the mempool where it can be seen by others triggering illegitimate claims. All claims pay the requisite fees and are confirmed to blocks in the specified time period, $\Delta$. The mechanism is then run drawing from confirmed claims. 

Given this, we can prove the following:

\begin{proposition}
Suppose that $\theta > 0$. The unique subgame perfect equilibrium involves a single claimant who is the agent who performs the obligation.
\end{proposition}

\begin{proof}
Without loss in generality, suppose that $A$ is the true claimant and the challenge state is initiated as $B$ also makes a claim. Thus, both agents have incurred the transaction fee, $f$. There are two cases to consider:
\begin{enumerate}
    \item If $A$ is given the opportunity to renounce their claim, they will receive $-\theta$ if they renounce their claim (as they know the other claimant is illegitimate) and $0$ otherwise. Thus, if $\theta > 0$, $A$ will continue to assert their claim and $T$ will be sent to a third party with each agent ultimately earning $-f$.
    \item If $B$ is given the opportunity to renounce their claim, they will receive $0$ regardless (as they know the other agent is legitimate). Thus, they will be indifferent between renouncing or not and their ultimate payoff will be $-f$.
\end{enumerate}
We now examine each agent's incentive to make a claim. As $A$ moves prior to $B$, we work backwards by considering $B$ choice. If $B$ makes a claim (by front-running), they expect to earn $-f$ as $A$ will never renounced their claim if $\theta > 0$. Thus, $B$ will not make a claim. Given that $B$ will not make a claim, $A$ will make a claim and earn $T - f$.
\end{proof}
\\

\noindent The fact that, to make a claim, agents must incur a transaction fee, $f$ makes this mechanism isomorphic to the mechanism considered for Solomon's dilemma where $F = f$. The only difference is that $f$ is not refunded to any claimant. 

Note that the mechanism, like that for Solomon's dilemma, requires that the true claimant have a strict preference regarding whether a payment is made to an illegitimate claimant. Otherwise, if there is a possibility that the legitimate claimant may renounce their claim, this opens up a potential return to illegitimate claimants. A weaker assumption that leads to this same outcome would be that if a true claimant is indifferent as to where the payment is made, if not to themselves, they choose to assert the claim. In equilibrium, if the true claimant were programming in their assertion response at the time they submit $M_A$, then it is optimal for them to assert their claim. Thus, the $\theta > 0$ assumption does not play role if agents precommit their mechanism responses as might arise in a Blockchain environment.

\begin{remark}
The mechanism also yields a single true claimant if there are many potential illegitimate claimants. The only difference is that one claimant out of the pool of claimants is given the opportunity to renounce their claim in the challenge stage. Because the true claimant knows they are the true claimant, they will also choose to assert their claim if $\theta > 0$.\footnote{One possible front-running strategy would be for a front-runner to send multiple messages for the same wallet address. To avoid this, a claim to be resolved would draw based on messages. Of course, front-runners could send messages for different wallets. This, however, would not exclude the true claimant and so would ultimately fail.} 
\end{remark}

\begin{remark}
There is a possibility that there could be multiple claimants with the illegitimate claimants forming a coalition. In this case, a randomly selected agent could chose to withdraw their claim but there still be multiple claimants. In this case, what would happen to the tokens? One way of overcoming this is, when there are more than 2 claimants, a set of agents are randomly selected. If any assert their claims, $T$ is paid to a third party. If all withdraw their claims, those agents are removed from the pool of claimants and a new set of agents (half or just under one half of the remainder) are randomly selected and the mechanism is repeated. Eventually, an agent who is asserting their claim will be selected and $T$ will be paid to a third party. There is guaranteed to be one such agent as the true claimant is amongst the starting pool.\footnote{If there were concerns that this process may take time, then a cost could be imposed on each round of participation.} 
\end{remark}

\begin{remark}
What if, due to network issues, the true claimant is not amongst the pool of claimants when the mechanism begins? If illegitimate claimants believe that this is a possibility, they have an incentive to make such claims. Clearly, if they are the only claimant, they will be able to capture $T$. If they are amongst multiple claimants, this is not possible unless, those multiple claimants are controlled by them. Thus, there would have to be some collusive mechanism amongst illegitimate claimants to subvert the mechanism in this way. So long as the probability that the legitimate claimant is not in the relevant pool at the time the mechanism is run is low enough, the deterrence effect of the mechanism remains.
\end{remark}

\begin{remark}
Since \cite{aghion2012subgame} it has been understood that certain mechanisms may not be robust to small perturbations from common knowledge. \cite{chen2018getting} show how suitably-designed lotteries can ensure that the mechanism used in this paper is robust to private-value perturbations. It would be straightforward to do so in this environment if desired, although it would make the mechanism slightly more involved.
\end{remark}

\begin{remark}
We could have specified a form of “mutually assured destruction” by automatically burning the tokens if there is more than one claimant. This is a reasonable approach, but our challenge stage permits the construction of lotteries mentioned in Remark 4 that make it the \textit{unique} optimal choice for an illegitimate claimant to withdraw their claim (rather than being indifferent). A practical feature of our mechanism is that an illegitimate claimant must pay an additional transaction fee to send the message at the challenge stage, as this must be written to the blockchain. Such a fee also breaks their indifference, while the true claimant has preferences that make them willing to pay a certain fee to assert at the challenge stage and ensure that an illegitimate claimant does not receive the tokens.
\end{remark}

\subsection{The multiple legitimate claimant case}

The above analysis envisages a contract on the blockchain where there is only one agent who can be the legitimate claimant. However, in some applications -- say involving bounties or rewards for performance -- can have multiple legitimate claimants. In the absence of front-runners, such rewards would be made based on some time verification of the messages from agents. That may result in multiple claimants but the contract could specify a tie-breaking rule or another measure to award the bounty including splitting it. When there are illegitimate claimants, however, those rules would create an incentive to front-run the contract. 

A potential solution in this case would be to run the mechanism as proposed for the one legitimate claimant case. This might be done by shortening the time ($\Delta$) where claims will be evaluated even if this results in potentially higher transaction fees. Reducing $\Delta$ means that any true claimant will more likely to believe that no other legitimate claimant will submit a competing claim during that period and there will be one true claimant. In that case, the fact that front-runners do not have an incentive to claim, will preserve the contractual incentives. The cost is that this will limit the tie-breaking options that might otherwise be used in such contracts. Such options are important if they play a role in providing incentives to compete and perform the contract obligations. 

There is, however, a counter-risk that arises in this particular case: the payor may have an incentive to front-run their own mechanism. This would arise if it could not be guaranteed that $T$ was being sent to a party other than the payor. Thus, the mechanism would have to specify that the tokens be burnt or sent to a legitimate charity account that is publicly verified. 

These potential issues, however, need to be weighed against the real possibility that the contract would be otherwise completely unworkable if front-running was possible. Thus, the use of the mechanism expands the feasible contract space but does not obtain the full range of options that would be available if front-running were not possible at all. 

\subsection{A note on implementation}

The mechanism we analyse here can be easily implemented on existing blockchains. Indeed, we have already provided code for a generic smart contract on the Ethereum network.\footnote{See the repository at \url{https://github.com/solomonic-mechanism}.} In effect, it is a Solomonic clause added to existing smart contracts. Figure \ref{fig:chart} shows a flow chart of its operation.

\begin{figure}
    \centering
       \caption{\textbf{Flowchart of Solomonic Clause}}
    \label{fig:chart}
    \includegraphics[scale=0.2]{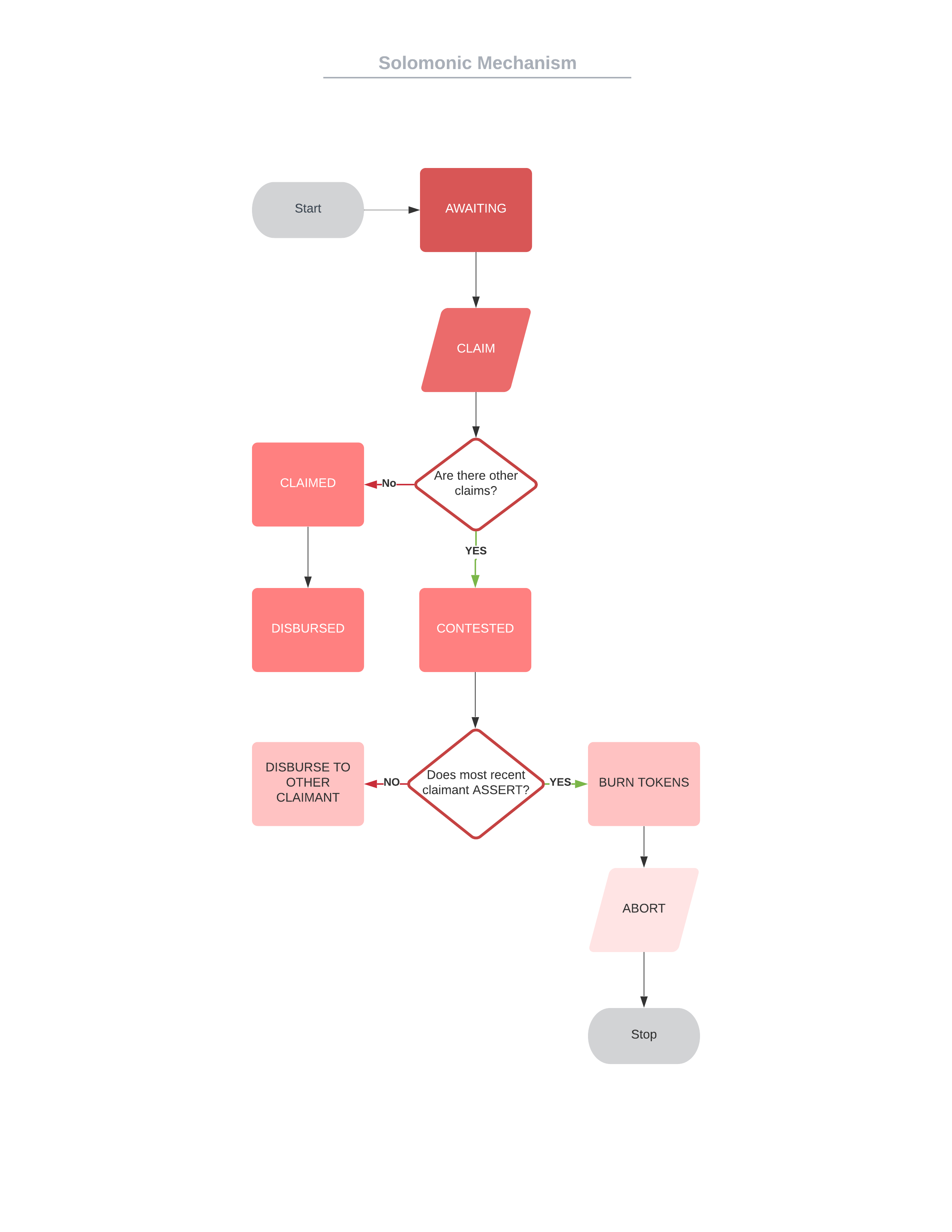}
\end{figure}

There are many open design choices in implementing Solomonic clauses that we list here but that their resolution is beyond the scope of the present paper. These include:
\begin{itemize}
    \item \textit{Hard-coded challenge response}: the mechanism as outlined includes a claim message followed by a message in the challenge stage if a potential claimant is selected. However, it could be envisaged that the initial message contains the contingent response in the challenge stage rather than requiring a separate message and fee payment.
    \item \textit{Randomization}: implementing randomization on a blockchain virtual machine can be challenging and often requires a call to an Oracle that is off-chain. In our implementation example, the agent chosen in the challenge stage was not chosen at random but was the agent with the most recent time-stamp on recorded on the blockchain. Theoretically, this the agent most likely to be the true claimant as they would not be putting forward higher transaction fees as part of a front-running strategy. However, due to latency on the internet, that agent is, in part, determined randomly and thus, this would be a useful alternative to pure randomization. 
    \item \textit{Time period}: in our implementation we set the time period, $\Delta$ to 60 seconds (or 4 blocks on Ethereum). The actual time period chosen would depend on other factors including network congestion and the need to clear token payments quickly or not. 
    \item \textit{Token burning}: If a claim is asserted in a challenge stage, then the tokens need to be transferred away from any party in the arrangement for the mechanism to work. This could involve burning the tokens (sending them to a null address) or, alternatively, having the tokens become part of a fund or non-profit. As the mechanism, if successful, should involve little of this in equilibrium, where the tokens are sent is a decision that should be made to ensure that the mechanism is not attacked by malicious agents trying to force the tokens to be burnt or otherwise undermine the operation of the smart contract.
    \item \textit{Signaling}: A contract with a Solomonic clause could involve a message for payment indistinguishable from contracts without that clause or one that indicated the existence of the Solomonic clause. The distinction would impact on front-running and its attempts. When there is a clear signal, front-runners would avoid these contracts. When there is no clear signal, they may not unless there was a sufficient share of contracts with a Solomonic clause in which case, front-running on all contracts may not be worthwhile. The use of such signals is, therefore, an important implementation choice.
\end{itemize}

\section{Conclusion}

We have outlined a mechanism which addresses front-running of smart contracts. An advantage of our approach is that the mechanism is embedded in a given smart contract, rather than needing to be deployed on the blockchain itself. The code we have provided shows how implement the mechanism in a smart contract, and we have trialed such contracts on the Ethereum blockchain (see \url{https://github.com/solomonic-mechanism} for details.) By removing a major impediment to smart contracting, we hope that such contracts will be able to achieve their potential, including the ability to write renegotiation-proof contracts that are not possible in traditional contracting environments.

Finally, it has not escaped our notice that the type of mechanism utilized here to address front-running can also be used as the cornerstone of a proof-of-stake consensus protocol, thus reducing transaction costs of achieving consensus and avoiding altogether the extreme energy use of proof-of-work protocols.

\newpage

\typeout{}
\bibliography{references}
\bibliographystyle{apalike}

\end{document}